\documentclass[conference]{IEEEtran}
\usepackage{amsthm}
\newtheorem{theorem}{Theorem}
\newtheorem{definition}{Definition}
\newtheorem{remark}{Remark}
\newtheorem{lemma}{Lemma}

\usepackage{ifpdf}

\usepackage{cite}

\ifCLASSINFOpdf
   \usepackage[pdftex]{graphicx}
\else
   \usepackage[dvips]{graphicx}
\fi
\usepackage{amssymb}

\usepackage[cmex10]{amsmath}
\usepackage{algorithmic}

\usepackage{array}
\usepackage{url}

\begin{document}
%
\title{Formation Stabilization with Collision Avoidance of Complex Systems}

\author{\IEEEauthorblockN{Soumic Sarkar}
\IEEEauthorblockA{Department of Electrical  Engineering\\
Indian Institute of Technology Delhi\\
Email: soumic4it@gmail.com}
\and
\IEEEauthorblockN{Indra Narayan Kar}
\IEEEauthorblockA{Department of Electrical  Engineering\\
Indian Institute of Technology Delhi\\
Email: ink@iitd.ac.in}}


%


\maketitle

\begin{abstract}
Two different aspects of formation control of multiple agents subjected to linear transformation have been addressed in this paper. We consider a set of complex single integrator systems so that the dimension of the system reduces to half as opposed to the vector representation in Cartesian coordinate system. We first design a stable formation controller in an attempt to solve the formation control turned to stabilization problem and then find a collision avoidance controller in the transformed domain, respectively. Different linear transformations are used to facilitate the formation control task in a different way. For example Jacobi transformation is used to separate the shape control and trajectory control. The inverse of the transformation must have nonzero eigenvalues with both positive and negative real parts which may lead the system to instability. If the inverse of the transformation appears in closed loop then a diagonal stabilizing matrix is required to reassign the eigenvalues of the inverse of transformation in the right half of complex plane. The algorithm to find such stabilizing matrix is provided. We then define a matrix of potential in the actual domain which is the stepping stone to find a matrix of potential in the transformed domain. Thus collision avoidance controller can be designed directly in the transformed domain. The mathematical proof is given that both the actual and transformed system behaves identically. Simulation results are provided to support our claim.
\end{abstract}


%
\IEEEpeerreviewmaketitle

\section{Introduction}
The study of the collective behaviour of birds, animals, fishes, etc. has not only drawn the attention of biologists, but also of computer scientists and roboticists. Thus several methods of cooperative control \cite{Saber:07} of multi-agent system have evolved, where a single robot is not sufficient to accomplish the given task, like navigation and foraging of unknown territory. These methods can broadly be categorized as the behaviour based approach (\cite{Reyn:87}, \cite{Balch:98}, \cite{Reif:99}), leader follower based approach (\cite{Desai:01}, \cite{Tanner:04}), virtual structure based approach (\cite{Leonard:01}, \cite{Lewis:97}, \cite{Eger:01}, \cite{Ren:04}), artificial potential based navigation (\cite{Gazi:05}, \cite{Per:08}, \cite{Zav:07}), graph theoretic method (\cite{Saber:07}, \cite{Saber:02}), formation shape control (\cite{Zhang:03}, \cite{Zhang:10}, \cite{Yang:12}, \cite{Silvia:11}). These strategies can either be centralized or decentralized depending upon the control laws designed.\\
In this paper we part the formation control problem for a set of complex single integrators subjected to real linear transformation into two sub-problems: stabilization and inter robot collision avoidance. Complex representation reduces the dimension of the system to half as opposed to Cartesian vector representation. The use of linear transformation to solicit formation control problems can be found in \cite{Zhang:03}, \cite{Silvia:11}. An example of such real linear transformation can be \textit{Centroid based Transformation} (CBT). Jacobi transformation \cite{Zhang:10} is one such CBT used to separate shape control and tracking of centroid. We are more interested to find a stabilizing matrix if the transformation appears in the closed loop system than to use the advantages of the transformation. Therefore we first rewrite the augmented system such that the inverse of linear transformation appears in the closed loop system. The eigenvalues of the inverse of the transformation must be nonzero but may have both positive and negative real parts. This may lead to instability of the combined system. Using a diagonal stabilizing matrix \cite{Fisher:58} we reassign the eigenvalues of the inverse of transformation in the right half of the complex plane so that the closed loop system becomes stable. We propose in this paper an alternative but systematic design procedure on multiplicative inverse eigenvalue problems (MIEPs)\cite{Chu:98} and \cite{Fried:75}. That is we update the inverse of the real CBT and reassign the eigenvalues by pre-multiplying a real diagonal matrix, called a \textit{stabilizing matrix} \cite{Fisher:58}. Sufficient conditions for the existence of a real stabilizing matrix are developed and algorithms are provided to find it.\\
The use of potential functions for navigation is not new to formation control community. One of the first works applying artificial potentials to agent coordination is found in \cite{Reif:99}, where they consider a distributed-control approach of groups of robots, called the \textit{social potential fields method}. The construction of artificial navigation functions along with collision avoidance is addressed in \cite{Rimon:92}. Two types of collision avoidance functions are there: obstacle avoidance \cite{Saber:06}, \cite{Ge:00}, \cite{Khatib:85} and inter robot collision avoidance \cite{Leonard:01}, \cite{Per:08}, \cite{Zav:07}, \cite{Dimaro:06}, \cite{Silvia:11}, \cite{Stip:07}. Obstacle avoidance functions are used to avoid static or moving obstacles, whereas inter robot collision functions are used to avoid collision among the robots while moving in formation. The diagrams of the negative gradient of collision avoidance functions are given in \cite{Khatib:85}, \cite{Zav:07}. These functions are all defined in the Cartesian Coordinate system based either on the distance between a robot and obstacles or the distance between a robot and its neighbours. The total potential of a robot is the summation of all the potentials due to the neighbouring robots. The potential term can be written in matrix form post multiplied with the states for the augmented system. We call this \textit{the matrix of potential}.\\
In this paper, we introduce a matrix of potential in transformed domain for a set of single integrator systems subjected to a linear transformation to study collision avoidance in formation control problem. Linear transformations to solicit formation control problem can be found in \cite{Zhang:03} and \cite{Silvia:11}. The use of potential function based controller in \cite{Zhang:03} and \cite{Zhang:10} may have restricted the usefulness of the transformation in the separation of shape control and trajectory tracking of the centroid. Whereas potential function based controller is used in \cite{Silvia:11}, but not in the transformed domain. This leads to the query what would be the collision avoidance controller in the transformed domain so that the stability analysis can be carried out entirely in the transformed domain. To address this issue, we have defined a matrix of potential for the augmented system in the Cartesian domain, which when post multiplied with states gives total effect of potential on the kinematics of a specific state. Then by using a linear transformation matrix we apply similarity transformation on the matrix of potential in Cartesian coordinate only to get the matrix of potential in the transformed domain. This matrix in hand, is then used to design the collision avoidance controller in the transformed domain. We have mathematically proved that the matrices of potential affect both the system in the actual and transformed domain in the same manner. \\
\textit{Notations:} $\mathbb{R}$ denotes the set of real numbers and $I_n$ denotes the identity matrix of order $n$. All elements of given matrices are multiplied with $I_2$ and therefore is omitted for brevity.

\hfill mds
 
\hfill December 27, 2012

\section{Priliminaries}

\subsection{Shape spaces}
Let $q_i=[x_i,y_i]^T\in\mathbb{R}^2, i=1,...,n$ denote the positions of a system of $n$ particles with respect to fixed inertial coordinate frame of reference. The positions of the same particles is denoted by $z_i=x_i+\iota y_i\in\mathbb{C}, i=1,...,n$ in complex coordinate system. Let there be another complex coordinate system $\xi=[\xi_i,\xi_c]^T, i=1,...,n-1$, where, $\xi_i=\sum_{i,j=1}^{n}p_{ij}z_i\in \mathbb{C}$ are \textit{shape vectors}, where $p_{ij}\in\mathbb{R}$, and are not all $0$s, and $\xi_c=\sum_{i=1}^n z_i\in\mathbb{C}$ denotes the \textit{centroid} of all positions. Then we define a real linear mapping for real systems $\mathbb{R}^{{2n}\times{2n}}:\mathbb{R}^n\rightarrow\mathbb{R}^n$ as
\begin{equation}
\label{rtrans}
T:q\rightarrow\xi
\end{equation}
And we also define a real linear mapping for complex systems $\mathbb{R}^{n\times n}:\mathbb{C}^n\rightarrow\mathbb{C}^n$ as
\begin{equation}
\label{ctrans}
T:z\rightarrow\xi
\end{equation}
Specific applications of such mapping $T$, $\mathbb{R}^{2n}\times\mathbb{R}^{2n}:\mathbb{R}^{2n}\rightarrow\mathbb{R}^{2n}$ can be found in \cite{Zhang:10} and \cite{Silvia:11}. For brevity we call the mapping $T$ as \textit{Centroid based transformation (CBT)}. One example of CBT is \textit{Jacobi transformation} to get \textit{Jacobi vectors} for \textit{Jacobi shape space}\cite{Zhang:03}. We will consider \textit{Jacobi vectors} as an example, to deduce our results in this article, though the results will comply similarly with other CBTs \cite{Silvia:11}. We give more stress to defining a new coordinate system to analyse the behaviour of the particles with respect to that reference frame, rather than investigating an interaction topology (communication among the agents), as in Graph theory \cite{Saber:07}, with respect to specific coordinate system (Cartesian coordinate or Inertial frame).
\begin{definition}
Let there be a real linear transformation $\mathbb{R}^{2n\times2n}:\mathbb{R}^{2n}\rightarrow\mathbb{R}^{2n}$ defined as
\begin{equation}
\label{realT}
T:x\rightarrow\xi
\end{equation}
and a complex weight matrix $W=diag\{w_i\}, i=1,...,n$ with $w_i=a_i+b_i\iota$, $a_i,b_i\in\mathbb{R}$. Then a complex linear transformation $\mathbb{C}^{2n\times2n}:\mathbb{R}^{2n}\rightarrow\mathbb{R}^{2n}$ can be defined as
\begin{equation}
\label{comT}
\Phi:x\rightarrow\xi
\end{equation}
where $\Phi=WT$.
\end{definition}

\section{Main Result}
We consider a group of $n$ agents in the plane labelled $1,...,n$. The positions of the $n$ agents are denoted by $x_1,...,x_n \in \mathbb{R}^2$, i.e.,$x_i=[p_{xi},p_{yi}]^T$. Then the positions of the same agents can be described in complex plane as $z_i=p_{xi}+\iota p_{yi}\in\mathbb{C}$, $i=1,...,n$. Suppose that each agent is governed by a single-integrator kinematics
\begin{equation}
\label{si}
\dot{z}_i=u_{fi}+u_{ai}
\end{equation}
where $z_i\in\mathbb{C}$ represents the position of $i^{th}$ agent in the complex plane and $u_{fi}=u_{fxi}+\iota u_{fyi}\in\mathbb{C}$ and $u_{ai}=u_{axi}+\iota u_{ayi}\in\mathbb{C}$ are the formation controller and the collision avoidance controller respectively. Then the combined dynamics of \eqref{si} can be written as
\begin{equation}
\label{combdyn}
\dot{z}=u_{f}+u_{a}
\end{equation}
where $z=[z_1,...,z_n]^T\in\mathbb{C}^{n}$, $u_f=[u_{f1},...,u_{fn}]^T\in\mathbb{C}^{n}$ and $u_a=[u_{a1},...,u_{an}]^T\in\mathbb{C}^{n}$. Consider that the system \eqref{combdyn} is driven by the following control laws
\begin{equation}
\label{sicntrl}
\begin{split}
u_{f}&=-K_zz_{e}+\dot{z}_{d} \\
u_{a}&=-\bigtriangledown f_{z}
\end{split}
\end{equation}
where $K_z\in \mathbb{C}^{n\times n}$ is the controller gain and $z_{e}=z-z_{d}\in \mathbb{C}^{n}$, with $z_{d}=[z_{d1},...,z_{d2}]^T\in \mathbb{C}^{n}$, and $z_{di}=z_{dxi}+\iota x_{dyi}\in\mathbb{C}$ is the desired trajectory given to the $i^{th}$ agent to follow. Define $\bigtriangledown f_i=\sum_{j\in N_i,j\neq i}\bigtriangledown f_{ij}\in \mathbb{C}$ with $\bigtriangledown f_{ij}\in \mathbb{C}$ is the gradient of any distance based potential function defined by $f_{ij}=f(z_i,z_j)$ with respect to $i^{th}$ agent. $z_i$ and $z_j$ are the positions of $i^{th}$ and $j^{th}$ agents. Let $z_{e}=[z_{e1},...,z_{en}]^T\in\mathbb{C}^{n}$ and 
$\bigtriangledown f_z=[\bigtriangledown f_1,...,\bigtriangledown f_n]^T\in\mathbb{C}^{n}$. Then the overall closed loop dynamics of $n$ agents can be written as
\begin{equation}
\label{clsdx}
\dot{z}_e=-K_zz_e-\bigtriangledown f_z
\end{equation}
Applying transformation of [] we get
\begin{equation}
\label{clsdxi}
\dot{\xi}_e=-K_{\xi}\xi_e-\bigtriangledown f_{\xi}
\end{equation}
where $K_{\xi}=\Phi K_z \Phi^{-1}$ and $\bigtriangledown f_{\xi}=\Phi\bigtriangledown f_x$. Let $K_z=\Phi^{-1}D$ with $D=diag\{d_1,...,d_n\}\in\mathbb{C}^{n\times n}$ be a diagonal stabilizing matrix []. Then $K_{\xi}=D\Phi^{-1}$ and the system \eqref{clsdxi} can be rewritten as
\begin{equation}
\label{clsdxiD}
\dot{\xi}_e=-D\Phi^{-1}\xi_e-\bigtriangledown f_{\xi}
\end{equation}


Next we define the tracking with formation problem.
\begin{definition}
The system \eqref{clsdxiD} asymptotically leads to a tracking in formation as $\lim_{t\rightarrow\infty}\xi_e=0$.
\end{definition}
\begin{theorem}
\label{thprob}
The system \eqref{clsdxiD} asymptotically leads to a tracking in formation as $\lim_{t\rightarrow\infty}\xi_e=0$ if and only if both $D\Phi^{-1}\xi_e=0$ and $\bigtriangledown f_{\xi}=0$ as $t\rightarrow\infty$.
\end{theorem}
\begin{proof}
\label{prfprob}
(Sufficiency). It is clear that if both $D\Phi^{-1}\xi_e=0$ and $\bigtriangledown f_{\xi}=0$ as $t\rightarrow\infty$, then $\lim_{t\rightarrow\infty}\xi_e=0$ by checking the dynamics of the system \eqref{clsdxiD}.
\end{proof}
\begin{remark}
From Theorem \ref{thprob}, it can be concluded that the tracking in formation problem parts into two sub-problems: Stabilization Problem and Collision Avoidance Problem, stated below.
\end{remark}
\begin{definition}
Let there be a system as follows
\begin{equation}
\label{stabprob}
\dot{y}=-D\Phi^{-1}y
\end{equation}
where $y\in\mathbb{C}^{n}$ the eigenvalues of the matrix $\Phi^{-1}$ are nonzero. Then find a stabilizing matrix $D\in\mathbb{C}^{n\times n}$ such that the eigenvalues of $D\Phi^{-1}$ have all positive real parts. Therefore $\lim_{t\rightarrow0}y=0$.
\end{definition}
\begin{definition}
Let there be a system as follows
\begin{equation}
\label{stabprob}
\dot{y}=-\bigtriangledown f
\end{equation}
where $\bigtriangledown f$ denotes the gradient of a potential function $f(y)$. Then find a collision avoidance function $f(y)$ such that $\lim_{t\rightarrow0}y=0$.
\end{definition}
We first solve the stabilization problem and then solicit the collision avoidance problem subsequently.
\subsection{Stabilization}
In this section we will discuss the Stabilization problem and leave the collision avoidance part for the next subsection. Hence, we remove $u_a$ from \eqref{combdyn} and apply transformation [] to get
\begin{equation}
\label{combxi}
\dot{\xi}=u_{\xi f}
\end{equation}
\begin{theorem}
The following controller 
\begin{equation}
\label{cntrlD}
u_{\xi f}=-D\Phi^{-1}\xi_e+\xi_{d}
\end{equation}
will stabilize the system \eqref{combxi}, i.e., $\lim_{t\rightarrow \infty} \xi=0$, if there exists a diagonal stabilizing matrix $D$ multiplied with $\Phi^{-1}$ all of whose leading principal minors are nonzero.
\end{theorem}
The proof requires the following result related to the multiplicative inverse eigenvalue problem.
\begin{theorem}
\label{thFishFull}
\cite{Fisher:58} Let $A$ be an $n\times n$ real matrix all of whose leading principal minors are positive. Then there is an $n\times n$ positive diagonal matrix $D$ such that all the roots of $DA$ are positive and simple.
\end{theorem}
\begin{proof}
\label{prfcntrlD}
By Theorem \ref{thFishFull}, it follows that there always exists a diagonal matrix $D$ multiplied with $\Phi^{-1}$ such that all the roots of $D\Phi^{-1}$ are positive and simple. From the definition of stability, if all the roots of $D\Phi^{-1}$ are positive, then $\lim_{t\rightarrow \infty} \xi=0$ with the control law \eqref{cntrlD} for the system \eqref{combxi}.
\end{proof}
We next give an algorithm to find the elements of $D$ iteratively using the leading principal minors of the matrix $\Phi^{-1}$. Before that we need to define a few quantities. Let the matrices $D$ and $\Phi^{-1}$ be partitioned as
\begin{equation}
\label{matpart}
\Phi^{-1}_{i+1}=\begin{bmatrix}
\Phi_{ii} & \Phi_{i1} \\
\Phi_{1i} & \Phi_{11}
\end{bmatrix} \ ; \ D_{i+1}=\begin{bmatrix}
D_{ii} & 0_{i1} \\
0_{1i} & d_{11}
\end{bmatrix}
\end{equation}
where $\Phi^{-1}_{i}\in\mathbb{C}^{i+1\times i+1}$, $i=1,...,n-1$ with $\phi_{11}\in\mathbb{C}$ (first principal minor of $\Phi^{-1}$) are the $n$ leading principal minors of the matrix $\Phi^{-1}$. $\Phi_{ii}\in\mathbb{C}^{i\times i}$, $\Phi_{i1}\in\mathbb{C}^{i\times 1}$, $\Phi_{1i}\in\mathbb{C}^{1\times i}$, and $\Phi_{11}\in\mathbb{C}^{1\times 1}$. The dimensions are similar for the diagonal stabilizing matrix $D_{i+1}$ with $D_{ii}=diag\{d_1,...,d_i\}\in\mathbb{C}^{i\times i}$. We then define another matrix $M(d_{i+1})=D_{i+1}\Phi_{i+1}$ as
\begin{equation}
\label{DA}
M(d_{i+1})=\begin{bmatrix}
D_{ii}\Phi_{ii} & D_{ii}\Phi_{i1} \\
d_{11}\Phi_{1i} & d_{11}\Phi_{11}
\end{bmatrix}
\end{equation}
Then the eigenvalues of the matrix $M(d_{i+1})$ can be calculated using the characteristic equation $det\left[\lambda I_{i+1}-M(d_{i+1})\right]=0$. The characteristic equation for the partitioned matrix $M(d_{i+1})$ can be written as
\begin{equation}
\label{chareq}
\begin{split}
&det(\lambda I_{i+1}-D_{ii}\Phi_{ii})*\\
&det((\lambda-d_{11}\Phi_{11})-d_{11}\Phi_{1i}\left[D_{ii}\Phi_{ii}\right]^{-1}D_{ii}\Phi_{i1})=0
\end{split}
\end{equation}
using Schur Compliments for partitioned matrices. \eqref{chareq} requires that $det(\lambda I_{i+1}-D_{ii}\Phi_{ii})\neq0$ which follows $det(D_{ii}\Phi_{ii})\neq 0$. This complies with the criteria that the leading principal minors of the matrix $\Phi^{-1}$ are nonzero. Note that if $D_{ii}$ is known, then the eigenvalues of \eqref{chareq} can be determined solely with the selection of $d_{11}$.\\
Next we give an algorithm to find $D$, in which the notation ${\Phi^{-1}}^{1\thicksim i}$ is used to denote the sub-matrix formed by the first $i$ rows and columns of $\Phi^{-1}$.\\
\textbf{Algorithm 1}\\
\rule{8.5cm}{.1mm}
\textbf{Input}: $\Phi^{-1}_{i+1}$ for $i=1,...,n-1$, with $\phi_{11}$.\\
\textbf{Output}: Stabilizing matrix $D$.\\
\textbf{Procedure}:\\
Select $\lambda^{'}_1>0$\\
Find $d_1$ from $d_1=\frac{\lambda^{'}_1}{\phi_{11}}$\\
\textbf{for} $i=1,...,n-1$ \textbf{do}\\
Find $d_{i+1}$ using \eqref{chareq} such that all the eigenvalues of $M(d_{i+1})$ are in the open right half of complex plane\\
\textbf{end for}\\
Construct $D=diag(d_1,...,d_n)$\\
\rule{8.5cm}{.1mm}\\
It is clear from the above algorithm that once $d_1$ is determined then only $d_2$ is unknown in equation \eqref{chareq} in the first iteration. Hence $d_2$ can be selected such that all the eigenvalues of $M(d_{2})$ are in the open right half of the complex plane. The process will repeat until all $d_{i+1}$ are found. One can thus conclude that when the matrix satisfies the assumption in Theorem \ref{thFishFull}, Algorithm 1 can always provide a solution $D$ without any exemption.

\subsection{Collision Avoidance}
Before presenting our results on stability, we provide a few definitions and theorems to facilitate our claim.
\begin{definition}
\label{def1}
Let there be a linear transformation $\Phi:z\rightarrow \xi$ between two non-empty configuration spaces $z=[z_1,...,z_n]\in\mathbb{C}^{n}$ and $\xi=[\xi_1,...,\xi_n]\in\mathbb{C}^{n}$, and its inverse $\Phi^{-1}:\xi\rightarrow z$ exists. Define $z_i=f_i^{-1}(\xi_1,...,\xi_n)$ and $\xi_i=f_i(z_1,...,z_n)$ for $i=1,...,n$ to denote that $x_i$ are linear combination of the vectors $\xi_1,...,\xi_n\in \xi$ and $\xi_i$ are linear combination of the vectors $z_1,...,z_n\in z$. Then the Euclidean distance between two points $z_i,z_j\in z$, $d_x\in\mathbb{R}$ is defined as follows
\begin{equation}
\label{edx}
d_z=\parallel z_{ij}\parallel=z_{ij}^{*}z_{ij}
\end{equation}
where $z_{ij}=z_i-z_j$, $(.)^{*}$ denotes the complex conjugate of a complex number and $z_{ij}^{*}$ is the complex conjugate of $z_{ij}$.\\
The Euclidean distance $d_z$ can also be expressed as a function of $\xi_1,...,\xi_n\in \xi$. Therefore write
\begin{equation}
\label{edxi}
d_z=d_{\xi}=\parallel [f_{ij}^{-1}]^{*}[f_{ij}^{-1}] \parallel
\end{equation}
where $f_{ij}^{-1}=f_i^{-1}(\xi_1,...,\xi_n)-f_j^{-1}(\xi_1,...,\xi_n)$ and $d_{\xi}$ refers to $d_z$ in terms of $\xi_1,...,\xi_n\in\xi$.
\end{definition}
\begin{remark}
The Euclidean distances of \eqref{edx} and \eqref{edxi} are equivalent.
\end{remark}
\textbf{Computation of Potential function in transformed domain :} \textit{Let there be a linear transformation $\Phi:z\rightarrow \xi$ between two non-empty configuration spaces $z=[z_1,...,z_n]\in\mathbb{C}^{n}$ and $\xi=[\xi_1,...,\xi_n]\in\mathbb{C}^{n}$, and its inverse $\Phi^{-1}:\xi\rightarrow z$ exists. Define $z_i=f_i^{-1}(\xi_1,...,\xi_n)$ and $\xi_i=f_i(z_1,...,z_n)$ for $i=1,...,n$ to denote that $z_i$ are linear combination of the vectors $\xi_1,...,\xi_n\in \xi$ and $\xi_i$ are linear combination of the vectors $z_1,...,z_n\in x$. Then a distance based potential function between two points $z_i,z_j\in z$, $\phi_z\in\mathbb{C}$ can be defined as
\begin{equation}
\label{pfx}
\phi_z=f_z(\parallel z_i-z_j\parallel)
\end{equation} 
where $\parallel z_i-z_j\parallel$ denotes the Euclidean distance among the agents $z_i\in x$ and $z_j\in z$ with $i\neq j$. The potential function $\phi_z$ can also be expressed as a function of $\xi_1,...,\xi_n\in \xi$. Therefore write
\begin{equation}
\label{pfxi}
\phi_{\xi}=f_{\xi}(\parallel f_i^{-1}(\xi_1,...,\xi_n)-f_j^{-1}(\xi_1,...,\xi_n)\parallel)
\end{equation}
where $\phi_{\xi}$ refers to $\phi_z$ in terms of $\xi_1,...,\xi_n\in\xi$ and $\parallel f_i^{-1}(\xi_1,...,\xi_n)-f_j^{-1}(\xi_1,...,\xi_n)\parallel$ denotes the Euclidean distance as in \eqref{edxi}.}

\begin{lemma}
\label{th1}
The potential functions of \eqref{pfx} and \eqref{pfxi} are equivalent.
\end{lemma}
\begin{proof}
\label{prf1}
(Sufficiency). Since by definition, $z_i=f_i^{-1}(\xi_1,...,\xi_n)$ and $z_j=f_j^{-1}(\xi_1,...,\xi_n)$, then $\parallel z_i-z_j\parallel=\parallel f_i^{-1}(\xi_1,...,\xi_n)-f_j^{-1}(\xi_1,...,\xi_n)\parallel$. Therefore, $f_z(\parallel z_i-z_j\parallel)=f_{\xi}(\parallel f_i^{-1}(\xi_1,...,\xi_n)-f_j^{-1}(\xi_1,...,\xi_n)\parallel)$. Hence, $\phi_z$ and $\phi_{\xi}$ are equivalent.\\
(Necessity). Suppose that $\phi_z$ and $\phi_{\xi}$ are not equivalent. Then there must exist a vector for which $z_i\neq f_i^{-1}(\xi_1,...,\xi_n)$ which is not possible by definition. Thus $\phi_z$ and $\phi_{\xi}$ are equivalent.
\end{proof}

\textbf{Computation of gradient:} \textit{Let there be a linear transformation $\Phi:z\rightarrow \xi$ between two non-empty configuration spaces $z=[z_1,...,z_n]\in\mathbb{C}^{n}$ and $\xi=[\xi_1,...,\xi_n]\in\mathbb{C}^{n}$, and its inverse $\Phi^{-1}:\xi\rightarrow z$ exists. Define $z_i=f_i^{-1}(\xi_1,...,\xi_n)$ and $\xi_i=f_i(z_1,...,z_n)$ for $i=1,...,n$ to denote that $z_i$ are linear combination of the vectors $\xi_1,...,\xi_n\in \xi$ and $\xi_i$ are linear combination of the vectors $z_1,...,z_n\in z$. Then the gradient of a distance based potential function \eqref{pfx} is defined as
\begin{equation}
\label{gradx}
\bigtriangledown_{ij}\phi_z=\bigtriangledown_{ij}f_z(\parallel z_i-z_j\parallel)
\end{equation} 
where $\parallel z_i-z_j\parallel$ denotes the Euclidean distance among the agents $z_i\in x$ and $z_j\in x$ with $i\neq j$. The gradient of potential function $\bigtriangledown_{ij}\phi_z$ can also be expressed as a function of $\xi_1,...,\xi_n\in \xi$. Therefore write 
\begin{equation}
\label{gradxi}
\bigtriangledown_{ij}\phi_{\xi}=\bigtriangledown_{ij}f_{\xi}(\parallel f_i^{-1}(\xi_1,...,\xi_n)-f_j^{-1}(\xi_1,...,\xi_n)\parallel)
\end{equation}
where $\bigtriangledown_{ij}\phi_{\xi}$ refers to $\bigtriangledown_{ij}\phi_x$ in terms of $\xi_1,...,\xi_n\in\xi$ and $\parallel f_i^{-1}(\xi_1,...,\xi_n)-f_j^{-1}(\xi_1,...,\xi_n)\parallel$ denotes the Euclidean distance as in \eqref{edxi}.}

\begin{theorem}
\label{th2}
The gradients of the potential function of \eqref{gradx} and \eqref{gradxi} are equivalent.
\end{theorem}
\begin{proof}
\label{prf2}
(Sufficieny). Since by definition, $z_i=f_i^{-1}(\xi_1,...,\xi_n)$ and $z_j=f_j^{-1}(\xi_1,...,\xi_n)$, then $\parallel z_i-z_j\parallel=\parallel f_i^{-1}(\xi_1,...,\xi_n)-f_j^{-1}(\xi_1,...,\xi_n)\parallel$. Therefore, $\bigtriangledown_{ij}f_z(\parallel z_i-z_j\parallel)=\bigtriangledown_{ij}f_{\xi}(\parallel f_i^{-1}(\xi_1,...,\xi_n)-f_j^{-1}(\xi_1,...,\xi_n)\parallel)$. Hence, $\bigtriangledown_{ij}\phi_z$ and $\bigtriangledown_{ij}\phi_{\xi}$ are equivalent.\\
(Necessity). Suppose that $\bigtriangledown_{ij}\phi_z$ and $\bigtriangledown_{ij}\phi_{\xi}$ are not equivalent. Then there must exist a vector for which $z_i\neq f_i^{-1}(\xi_1,...,\xi_n)$ which is not possible by definition. Thus $\phi_z$ and $\phi_{\xi}$ are equivalent.
\end{proof}

\begin{lemma}
\label{def4}
Let $\bigtriangledown_{ij}\phi_z$ of \eqref{gradx} is defined as
\begin{equation}
\label{fxdef}
\bigtriangledown_{ij}f_z(\parallel z_i-z_j\parallel)=p_{z_{ij}}(z_i-z_j)
\end{equation}
where $p_{z_{ij}}=f_{z_{ij}}(\parallel z_i-z_j\parallel)\in\mathbb{C}$. Then the total potential of an agent $z_i\in z$ due to all neighbouring agents $z_j\neq z_i\in z$ can be defined as
\begin{equation}
\label{tpx}
\bigtriangledown\phi_{zi}=\sum_{j=1,j\neq i}^{N}p_{z_{ij}}(z_i-z_j)
\end{equation}
Then the total potential for agents $i=1,...,n$ can be written in augmented form as
\begin{equation}
\label{augpx}
\bigtriangledown\Phi_{z}=P_zz
\end{equation}
where the elements of the matrix $P_z$, $p_{ij}=\sum_{j=1,j\neq i}^{N}p_{z_{ij}}$. We then define the total potential in $\xi\in\mathbb{C}^{n}$ for $i=1,...,n$ in augmented form as
\begin{equation}
\label{augpxi}
\bigtriangledown \Phi_{\xi}=\Phi P_{\xi}\Phi^{-1}\xi
\end{equation}
where $P_{\xi}=P_z$ [see next \textit{Theorem}].
\end{lemma}

\begin{theorem}
\label{th5}
Suppose that there is a system of equations in $x$ defined as
\begin{equation}
\label{sysx}
\dot{x}=-\bigtriangledown\Phi_{z}
\end{equation}
and another system of equations in $\xi$ defined as
\begin{equation}
\label{sysxi}
\dot{\xi}=-\bigtriangledown\Phi_{\xi}
\end{equation}
Then \eqref{sysx} and \eqref{sysxi} are equivalent under the transformation $\Phi:z\rightarrow \xi$.
\end{theorem}
\begin{proof}
\label{prf5}
By multiplying \eqref{sysx} with $\Phi$ and using \eqref{augpx} , we get
\begin{equation}
\label{step1}
\dot{\xi}=-\Phi P_z \Phi^{-1} \xi
\end{equation}
But $P_z=P_{\xi}$ by \textit{Theorem} \ref{th4}. Also the null space of $P_z$ is the same as the null space of $\Phi P_z \Phi^{-1}$. Hence, \eqref{sysx} and \eqref{sysxi} are equivalent.
\end{proof}
\begin{remark}
Note that every equilibrium set of system \eqref{sysxi} is equal to every equilibrium set of system \eqref{sysx}, as per \textit{Theorem} \ref{th5}.
\end{remark}
\begin{theorem}
\label{th6}
Suppose that there is a system of equations in $x$ defined as
\begin{equation}
\label{sys2x}
\dot{z}=-K_zz
\end{equation}
and another system of equations in $\xi$ defined as
\begin{equation}
\label{sys2xi}
\dot{\xi}=-K_{\xi} \xi
\end{equation}
where $K_{\xi}=\Phi K_z \Phi^{-1}$. Then \eqref{sys2x} and \eqref{sys2xi} are equivalent under the transformation $\Phi:z\rightarrow \xi$.
\end{theorem}
\begin{proof}
\label{prf6}
As the null space of $K_z$ is the same as the null space of $K_{\xi}$, the two systems \eqref{sys2x} and \eqref{sys2xi} have the same equilibrium set and therefore are equivalent.
\end{proof}
\begin{theorem}
\label{th6}
Suppose that there is a system of equations in $z$ defined as
\begin{equation}
\label{sys3x}
\dot{z}_e=-K_zz_e-P_zz
\end{equation}
and another system of equations in $\xi$ defined as
\begin{equation}
\label{sys3xi}
\dot{\xi}_e=-\Phi K_z\Phi^{-1}\xi_e-\Phi P_z \Phi^{-1} \xi
\end{equation}
Then \eqref{sys3x} and \eqref{sys3xi} are equivalent under the transformation $\Phi:z\rightarrow \xi$.
\end{theorem}
\begin{proof}
\label{prf6}
As the null space of $-K_zz_e-P_zz$ is the same as the null space of $-\Phi K_z\Phi^{-1}\xi_e-\Phi P_z \Phi^{-1} \xi$, the two systems \eqref{sys3x} and \eqref{sys3xi} have the same equilibrium set and therefore are equivalent.
\end{proof}
\begin{theorem}
\label{th8}
$\lim_{t\rightarrow \infty} \xi_e=0$ for the system 
\begin{equation}
\label{sys}
\dot{\xi}=u
\end{equation}
with the control law
\begin{equation}
\label{cntrl}
u=-K_{\xi}\xi_e+\dot{\xi}_d-\bigtriangledown\Phi_{\xi}
\end{equation}
\end{theorem}
\begin{proof}
\label{prf8}
The control law of \eqref{cntrl} will give the closed loop dynamics of \eqref{clsdxi}. As $p_{z_{ij}}\rightarrow 0$ as $t\rightarrow \infty$, the closed loop dynamics of \eqref{clsdxi} will reduce to
\begin{equation}
\label{clsdwp}
\dot{\xi}_e=-K_{\xi}\xi_e
\end{equation}
as $P_{\xi}\rightarrow 0$ or $\bigtriangledown\Phi_{\xi}\rightarrow 0$. Then if the eigenvalues of the matrix $K_{\xi}$ are in the right half of s-plane, then $\lim_{t\rightarrow \infty} \xi_e=0$.
\end{proof}
Suppose that the Jacobi vectors of [] for $n=3$ agents are given as
\begin{equation}
\label{jacgen}
\begin{split}
\xi_1&=\sqrt{\mu_1}(z_2-z_1) \\
\xi_2&=\sqrt{\mu_2}\bigg(z_3-\frac{m_1z_1+m_2z_2}{m_1+m_2}\bigg ) \\
\xi_c&=\frac{\sum_{i=1}^{3}m_ix_i}{M}
\end{split}
\end{equation}
where $M=\sum_{i=1}^{n}m_i$ and
\begin{equation}
\label{mu}
\frac{1}{\mu_i}=\frac{1}{\sum_{k=1}^i}+\frac{1}{m_{i+1}} \ for \ i=1,2
\end{equation}
The Jacobi vectors can be rewritten as
\begin{equation}
\label{jacnorm}
\begin{split}
\xi_1&=a_{11}z_1+a_{12}z_2 \\
\xi_2&=a_{21}z_1+a_{22}z_2+a_{23}z_3 \\
\xi_c&=a_{31}z_1+a_{32}z_2+a_{33}z_3
\end{split}
\end{equation}
where $a_{11}=-\mu_1$, $a_{12}=\mu_1$, $a_{21}=-\frac{m_1\sqrt{\mu_2}}{m_1+m_2}$, $a_{21}=-\frac{m_2\sqrt{\mu_2}}{m_1+m_2}$, $a_{23}=\sqrt{\mu_2}$, $a_{31}=\frac{m_1}{M}$, $a_{32}=\frac{m_2}{M}$ and $a_{33}=\frac{m_3}{M}$. Then equation \eqref{jacnorm} can be written in augmented form as
\begin{equation}
\label{jacmat}
\begin{bmatrix}
\xi_1 \\ \xi_2 \\ \xi_c
\end{bmatrix}=\begin{bmatrix}
a_{11} & a_{12} & 0 \\
a_{21} & a_{22} & a_{23} \\
a_{31} & a_{32} & a_{33}
\end{bmatrix}\begin{bmatrix}
z_1 \\ z_2 \\ z_3
\end{bmatrix}
\end{equation}
Hence, the Jacobi Transformation for $n=3$ is 
\begin{equation}
\label{jactrans}
\Phi=\begin{bmatrix}
a_{11} & a_{12} & 0 \\
a_{21} & a_{22} & a_{23} \\
a_{31} & a_{32} & a_{33}
\end{bmatrix}
\end{equation}
The vectors in Cartesian Coordinate can be expressed as
\begin{equation}
\label{cvec}
\begin{split}
z_1=A_{11}\xi_1+A_{12}\xi_2+A_{13}\xi_c \\
z_2=A_{21}\xi_1+A_{22}\xi_2+A_{23}\xi_c \\
z_3=A_{31}\xi_1+A_{32}\xi_2+A_{33}\xi_c
\end{split}
\end{equation}
where
\begin{equation}
\label{Aval}
\begin{split}
A_{11}&=\left[ \frac{1}{a_{11}}-\frac{a_{12}}{A_{x_2}}(\frac{a_{23}a_{31}}{a_{11}}-\frac{a_{33}a_{21}}{a_{11}})\right] \\
A_{12}&=\frac{a_{12}a_{33}}{A_{x_2}}; \ A_{13}=\frac{a_{12}a_{23}}{A_{x_2}} \\
A_{21}&=\frac{1}{A_{x_2}}(\frac{a_{23}a_{31}}{a_{11}}-\frac{a_{33}a_{21}}{a_{11}}) \\
A_{22}&=\frac{a_{33}}{A_{x_2}}; \ A_{23}=-\frac{a_{23}}{A_{x_2}} \\
A_{31}&=\frac{1}{A_{x_3}}\left[ \frac{a_{31}}{a_{11}}(a_{22}-\frac{a_{21}a_{12}}{a_{11}})-\frac{a_{21}}{a_{11}}(a_{32}-\frac{a_{31}a_{12}}{a_{11}})\right] \\
A_{32}&=\frac{1}{A_{x_3}}\left(a_{32}-\frac{a_{31}a_{12}}{a_{11}}\right); \\
A_{33}&=-\frac{1}{A_{x_3}}\left(a_{22}-\frac{a_{21}a_{12}}{a_{11}}\right) \\
A_{x_2}&=a_{33}\left(a_{22}-\frac{a_{21}a_{12}}{a_{11}}\right)-a_{23}\left( a_{32}-\frac{a_{31}a_{12}}{a_{11}}\right) \\
A_{x_3}&=a_{23}\left(a_{32}-\frac{a_{31}a_{12}}{a_{11}}\right)-a_{33}\left( a_{22}-\frac{a_{21}a_{12}}{a_{11}}\right)
\end{split}
\end{equation}
Then equation \eqref{cvec} can be written is augmented form as
\begin{equation}
\label{invjacmat}
\begin{bmatrix}
z_1 \\ z_2 \\ z_3
\end{bmatrix}=\begin{bmatrix}
A_{11} & A_{12} & A_{13} \\
A_{21} & A_{22} & A_{23} \\
A_{31} & A_{32} & A_{33}
\end{bmatrix}\begin{bmatrix}
\xi_1 \\ \xi_2 \\ \xi_c
\end{bmatrix}
\end{equation}
Hence the inverse of Jacobi Transformation is
\begin{equation}
\label{jactrans}
\Phi^{-1}=\begin{bmatrix}
A_{11} & A_{12} & A_{13} \\
A_{21} & A_{22} & A_{23} \\
A_{31} & A_{32} & A_{33}
\end{bmatrix}
\end{equation}
Let us choose a potential function of [] as
\begin{equation}
\label{potcart}
v_{ij}(z_i,z_j)=\left(min\left\{0,\frac{\parallel z_i-z_j\parallel^2-R^2}{\parallel z_i-z_j\parallel^2-R^2}\right\}\right)^2
\end{equation}
where $R>r>0$. $z_i,z_j$ denotes the position of robots and $v_{ij}(z_i,z_j)$ denotes the potential function of $z_i$ with respect to the neighbouring robot $z_j$. $R$ denotes the radius of detection and $r$ is the avoidance radius. Denote $z_{ij}=z_i-z_j$. Then the gradient of the potential function with respect to $z_{ij}$ is
\begin{equation}
\label{gradc}
\frac{\partial v_{ij}}{z_{ij}}= 
\begin{cases}
0, & \text{if $\parallel z_i-z_j\parallel\geq R$}\\
p_{z_{ij}}\left(z_i-z_j\right), & \text{if $R\geq\parallel x_i-x_j\parallel\geq r$}\\
\text{not defined}, & \text{if $\parallel z_i-z_j\parallel= r$}\\
0, & \text{if $\parallel z_i-z_j\parallel< r$}\\
\end{cases}
\end{equation}
where $p_{z_{ij}}=4\frac{\left(R^2-r^2\right)\left(\parallel z_i-z_j\parallel^2-R^2\right)}{\left(\parallel z_i-z_j\parallel^2-r^2\right)^3}$. Denote $x_i=f_i(z_1,z_2,z_c)$ and $x_j=f_j(z_1,z_2,z_c)$. For brevity, we write $f_i$ and $f_j$ to denote $f_i(z_1,z_2,z_c)$ and $f_j(z_1,z_2,z_c)$ respectively. Define $f_{ij}=f_i-f_j$. Then $p_{f_{ij}}=4\frac{\left(R^2-r^2\right)\left(\parallel f_i-f_j\parallel^2-R^2\right)}{\left(\parallel f_i-f_j\parallel^2-r^2\right)^3}$. From the collision avoidance law for $i^{th}$ agent
\begin{equation}
\label{totpotc}
u_{ai}=-\sum_{j=1,j\neq i}^{n}\frac{\partial v_{ij}}{x_{ij}}
\end{equation}
we have the following control laws for $n=3$ agents
\begin{equation}
\label{cntrlcart}
\begin{split}
u_{a1}&=(p_{f_{12}}+p_{f_{13}})z_1-p_{f_{12}}z_2-p_{f_{13}}z_3 \\
u_{a2}&=-p_{f_{21}}z_1+(p_{f_{23}}+p_{f_{21}})z_2-p_{f_{23}}z_3 \\
u_{a3}&=-p_{f_{31}}z_1-p_{f_{32}}z_2+(p_{f_{31}}+p_{f_{32}})z_3
\end{split}
\end{equation}
The equation \eqref{cntrlcart} can be written in augmented form as
\begin{equation}
\label{augcntrlcart}
\begin{bmatrix}
u_{a1} \\ u_{a2} \\ u_{a3}
\end{bmatrix}=\begin{bmatrix}
(p_{f_{12}}+p_{f_{13}}) & -p_{f_{12}} & -p_{f_{13}} \\
-p_{f_{21}} & (p_{f_{23}}+p_{f_{21}}) & -p_{f_{23}} \\
-p_{f_{31}} & -p_{f_{32}} & (p_{f_{31}}+p_{f_{32}})
\end{bmatrix}\begin{bmatrix}
z_1 \\ z_2 \\ z_3
\end{bmatrix}
\end{equation}
Then we write
\begin{equation}
\label{acntrlcart}
u_a=p_fz
\end{equation}
where
\begin{equation}
\label{potmatcart}
p_f=\begin{bmatrix}
(p_{f_{12}}+p_{f_{13}}) & -p_{f_{12}} & -p_{f_{13}} \\
-p_{f_{21}} & (p_{f_{23}}+p_{f_{21}}) & -p_{f_{23}} \\
-p_{f_{31}} & -p_{f_{32}} & (p_{f_{31}}+p_{f_{32}})
\end{bmatrix}
\end{equation}
We write
\begin{equation}
\label{cntrlc}
u_a=p_f\Phi^{-1}\xi
\end{equation}
Note that $u_a=f(\xi)$.
Next we define the potential function in transformed domain as
\begin{equation}
\label{potinjac}
u_{\xi a}=\Phi p_f\Phi^{-1}\xi=p_{\xi}\xi
\end{equation}
where $p_{\xi}=\Phi p_f\Phi^{-1}$. Note that $\xi_i=f_i(z_1,...,z_n)$. Therefore the summation of the gradient of the potential functions with $\xi_i$ must associate the linear combination of the potentials due to $z_1,...,z_n\in z$ as in \eqref{potinjac}, obviously, represented in terms of $\xi_1,...,\xi_n\in \xi$.

\subsection{Double integrator dynamics}
Suppose that each agent is driven by a double-integrator dynamics
\begin{equation}
\label{did}
\ddot{z}_i=v_{fi}
\end{equation}
where the position $z_i\in\mathbb{C}$ and the velocity $\dot{z}_i\in\mathbb{C}$ are the states, and the acceleration $v_{fi}\in\mathbb{C}$ is the control input. Then the combined dynamics of \eqref{did} can be written as
\begin{equation}
\label{comdid}
\ddot{z}=v_f
\end{equation}
where $z=[z_1,...,z_n]^T\in\mathbb{C}^n$ and $v_f=[v_{f1},...,v_{fn}]^T\in\mathbb{C}^n$. Consider that system \eqref{comdid} is driven by the following control law
\begin{equation}
\label{didcntrl}
v_f=-K_pz_e-K_v\dot{z}_e+\ddot{z}_d
\end{equation}
where $K_p,K_v\in \mathbb{C}^{n\times n}$ are the controller gains and $z_{e}=z-z_{d}\in \mathbb{C}^{n}$, with $z_{d}=[z_{d1},...,z_{d2}]^T\in \mathbb{C}^{n}$, and $z_{di}=z_{dxi}+\iota x_{dyi}\in\mathbb{C}$ is the desired trajectory given to the $i^{th}$ agent to follow. Then the overall closed loop system of double integrators can be written in the following compact form
\begin{equation}
\label{didclsd}
\ddot{z}_e=-K_pz_e-K_v\dot{z}_e
\end{equation}
Applying transformation \eqref{comT} in \eqref{didcntrl} we get
\begin{equation}
\label{didcntrlT}
\ddot{\xi}_e=-K_{p\xi}\xi_e-K_{v\xi}\dot{\xi}_e
\end{equation}
where $K_{p\xi}=\Phi K_p\Phi^{-1}$ and $K_{v\xi}=\Phi K_v\Phi^{-1}$. Let $K_p=\Phi^{-1}D_1$ and $K_v=\Phi^{-1}D_2$. Then the combined system of \eqref{didcntrlT} can be written in matrix form as
\begin{equation}
\label{didMat}
\begin{bmatrix}
\dot{\xi} \\ \ddot{\xi}
\end{bmatrix}=\begin{bmatrix}
\mathbf{0} & \mathbf{I} \\
-D_1\Phi^{-1} & -D_2\Phi^{-1}
\end{bmatrix}\begin{bmatrix}
\xi \\ \dot{\xi}
\end{bmatrix}
\end{equation}
\begin{theorem}
\label{didTh}
Let there be a complex transformation $\Phi:z\rightarrow \xi$ and its inverse exists. Also the leading principal minors of $\Phi^{-1}$ are nonzero. Then there exist two diagonal stabilizing matrices $D_1$ and $D_2$ for the system \eqref{didMat}.

\end{theorem}
\begin{proof}
Denote 
\begin{equation}
\label{didSM}
A=\begin{bmatrix}
\mathbf{0} & \mathbf{I} \\
-D_1\Phi^{-1} & -D_2\Phi^{-1}
\end{bmatrix}
\end{equation}
Let $[x^T,y^T]^T$ be the right eigenvector of $A$. We then have
\begin{equation}
\label{vect}
A\begin{bmatrix}
x \\ y
\end{bmatrix}=\lambda\begin{bmatrix}
x \\ y
\end{bmatrix}
\end{equation}
where $\lambda_i$ is an eigenvalue of $A$ corresponding to the eigenvector $[x^T,y^T]^T$. \eqref{vect} implies the following equations
\begin{equation}
\label{eivec1}
y=\lambda_ix
\end{equation}
\begin{equation}
\label{eivec2}
-D_1\Phi^{-1}x-D_2\Phi^{-1}y=\lambda_iy
\end{equation}
Using \eqref{eivec1}, \eqref{eivec2} is written as
\begin{equation}
\label{intrm1}
-D_1\Phi^{-1}x-\lambda_iD_2\Phi^{-1}x=\lambda_i^2x
\end{equation}
The matrices $D_1\Phi^{-1}$ and $D_2\Phi^{-1}$ share same set of eigenvector $x$ and is therefore simultaneously diagonalizable. Let $D_1\Phi^{-1}x=\sigma_{1i}x$ and $D_2\Phi^{-1}x=\sigma_{2i}x$, where $\sigma_{1i}$ and $\sigma_{2i}$ are complex eigenvalues of the matrices $D_1\Phi^{-1}$ and $D_2\Phi^{-1}$. Since $D_1\Phi^{-1}$ and $D_2\Phi^{-1}$ are simultaneously diagonalizable the following is true
\begin{equation}
\label{assump}
\sigma_{2i}=\gamma\sigma_{1i}
\end{equation}
Replacing \eqref{assump} to get
\begin{equation}
\label{intrm2}
-\sigma_{1i}x-\gamma\sigma_{1i}\lambda_ix=\lambda_i^2x
\end{equation}
From \eqref{intrm2}, we have the following second order characteristic polynomial
\begin{equation}
\label{charpol}
\lambda_i^2+\gamma\sigma_{1i}\lambda_i+\sigma_{1i}=0
\end{equation}
Thus to show the existence of stabilizing matrices $D_1$ and $D_2$, it remains to show that $\sigma_{1i} (1,...,n)$ can be assigned such that the roots of the complex coefficient characteristic equation \eqref{intrm2} have negative real parts. \\
Define $a_1=a_2=Re(\sigma_{1i})$ and $b_1=b_2=Im(\sigma_{1i})$. Then from \cite{Xie:85} a second order polynomial with complex coefficients can be written as follows
\begin{equation}
\label{order2}
F_2(\lambda)=\lambda^2+(a_1+ib_1)\lambda+(a_2+ib_2)
\end{equation}
By $\mathit{Theorem 2}$ of \cite{Xie:85}, \eqref{order2} is stable if and only if the following polynomial of order 1 is stable
\begin{equation}
\label{order1}
F_1(\lambda)=a_1\lambda+a_1a_2^1+b_1^1b_2+ib_2
\end{equation}
where $b_1^1=a_1b_1-b_2$ and $a_2^1=a_1a_2$. Then the root of $F_1(\lambda)$ in \eqref{order1} is as follows
\begin{equation}
\label{rootofF1}
\lambda=-\frac{a_1^2a_2+(a_1b_1-b_2)b_2}{a_1}-\frac{b_1}{a_1}
\end{equation}
The real part of the root in \eqref{rootofF1} holds the following inequality for \eqref{order1} to be stable
\begin{equation}
\label{ineq2}
\frac{a_1^2a_2+a_1b_1b_2-b_2^2}{a_1}>0
\end{equation}
where $a_1>0$. Replacing $a_1,a_2,b_1,b_2$ we get the following condition for the existence of $D_1$ and $D_2$ and to prove the theorem

\begin{equation}
\label{ineqprf}
\frac{Re(\sigma_{1i})^3}{Im(\sigma_{1i})^2[1-Re(\sigma_{1i})]}>\frac{1}{\gamma^2}
\end{equation}

\end{proof}

\begin{remark}
Under the assumption of \eqref{assump} and the condition \eqref{ineqprf}, there exist $D_1$ and $D_2$ such that system \eqref{didMat} is stable. Therefore it is sufficient to find a $D_1$, as the eigenvalues of $D_2$ are just a scalar multiple of the eigenvalues of $D_1$.
\end{remark}

From the proof of \textit{Theorem 9}, we find that stabilizing matrices $D_1$ and $D_2$ can be found from for double-integrator dynamics including the condition \eqref{ineqprf} in the \textbf{Algorithm 1}. We present it below.\\

\textbf{Algorithm 2}\\
\rule{8.5cm}{.1mm}
\textbf{Input}: $\Phi^{-1}_{i+1}$ for $i=1,...,n-1$, with $\phi_{11}$.\\
\textbf{Output}: Stabilizing matrix $D_1$ and $D_2$.\\
\textbf{Procedure}:\\
Select $\lambda^{'}_1>0$\\
Find $d_1$ from $d_1=\frac{\lambda^{'}_1}{\phi_{11}}$\\
\textbf{for} $i=1,...,n-1$ \textbf{do}\\
Find $d_{i+1}$ using \eqref{chareq} such that all the eigenvalues of $M(d_{i+1})$ are in the open right half of complex plane\\
\textbf{end for}\\
Select a $D$ that satisfies the condition \eqref{ineqprf}\\
Construct $D=diag(d_1,...,d_n)$\\
Choose a scalar $\gamma$\\
Construct $D_1=D$ and $D_2=\gamma D$\\
\rule{8.5cm}{.1mm}\\

\begin{theorem}
Apply the transformation \eqref{ctrans} on \eqref{comdid} to get the following system
\begin{equation}
\label{comdideq}
\ddot{\xi}=v_{\xi f}
\end{equation}
Then the following control law 
\begin{equation}
\label{didcntrlxi}
v_{\xi f}=-D_1\Phi^{-1}\xi_e-D_2\Phi^{-1}\dot{\xi}_e+\ddot{\xi}_d
\end{equation}
will stabilize the system \eqref{comdideq}, if there exist two diagonal stabilizing matrices $D_1$ and $D_2$ for $\Phi^{-1}$, all of whose leading principal minors are nonzero.
\end{theorem}
\begin{proof}
From \textit{Theorem 9}, it can be checked that under the condition stated in \eqref{ineqprf}, there exist stabilizing matrices $D_1$ and $D_2$ such that all of the eigenvalues of $A$ lies in left half of complex plane. Therefore the proposed control law \eqref{didcntrlxi} will stabilize the system \eqref{comdideq}.
\end{proof}

\begin{theorem}
\label{th11}
$\lim_{t\rightarrow \infty} \xi_e=0$ for the system 
\begin{equation}
\label{sys2}
\dot{\xi}=u
\end{equation}
with the control law
\begin{equation}
\label{cntrl2}
u=-K_{p\xi}\xi_e-K_{v\xi}\dot{\xi}_e+\dot{\xi}_d-\bigtriangledown\Phi_{\xi}
\end{equation}
\end{theorem}
\begin{proof}
\label{prf11}
The control law of \eqref{cntrl2} will give the closed loop dynamics of \eqref{didclsd}. As $p_{z_{ij}}\rightarrow 0$ as $t\rightarrow \infty$, the closed loop dynamics of \eqref{didclsd} will reduce to
\begin{equation}
\label{clsdwp2}
\dot{\xi}_e=-K_{p\xi}\xi_e-K_{v\xi}\dot{\xi}_e
\end{equation}
as $P_{\xi}\rightarrow 0$ or $\bigtriangledown\Phi_{\xi}\rightarrow 0$. Then if the eigenvalues of the matrix $K_{p\xi}=D_1\Phi^{-1}$ and $K_{p\xi}=D_1\Phi^{-1}$ are selected by choosing $D_1$ and $D_2$ as per the condition stated in \textit{Theorem 10}, such that the eigenvalues of matrix $A$ of \eqref{didMat} lie in the left half of s-plane, then the control law \eqref{cntrl2} will stabilize the system \eqref{sys2}
\end{proof}

\section{Simulation Results}
\begin{figure}[h]
\begin{center}
\includegraphics[width=8cm, height=6cm]{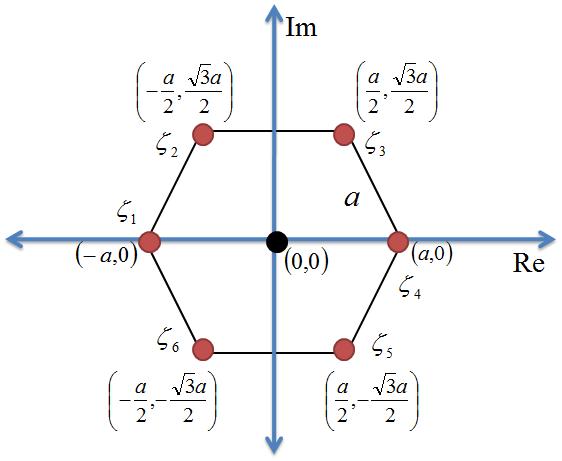}    
\caption{Gaius Julius Caesar, 100--44 B.C.}  
\label{fig1}                                 
\end{center}                                 
\end{figure}
The simulations are carried out on $6$ agents with single integrator kinematics of real and complex dynamics. The details of each simulation is given in the following subsections. The Jacobi vectors of $6$ agents of real and complex dynamics \cite{Yang:12} are given below
\begin{equation}
\label{jack7}
\begin{split}
\xi_1&=\mu_1(x_2-x_1) \ ; \
\xi_2=\mu_2[x_3-\frac{1}{2}(x_1+x_2)] \\
\xi_3&=\mu_3[x_4-\frac{1}{3}(x_+x_2+x_3)] \\
\xi_4&=\mu_4(x_5-\frac{1}{4}(x_1+x_2+x_3+x_4)] \\
\xi_5&=\mu_5(x_6-\frac{1}{5}(x_1+x_2+x_3+x_4+x_5)] \\
\xi_c&=\frac{1}{6}(x_1+x_2+x_3+x_4+x_5+x_6)
\end{split}
\end{equation}
where, $\mu_1=\frac{1}{\sqrt{2}}$, $\mu_2=\frac{\sqrt{2}}{\sqrt{3}}$, $\mu_3=\frac{\sqrt{3}}{2}$, $\mu_4=\frac{2}{\sqrt{5}}$, $\mu_5=\frac{\sqrt{5}}{\sqrt{6}}$.\\
Consider a planner formation with a formation basis defined as $z_d=a[(-1,0)$, $(-\frac{1}{2},\frac{\sqrt{3}}{2})$, $(\frac{1}{2}, \frac{\sqrt{3}}{2})$, $(1,0)$, $(\frac{1}{2}, -\frac{\sqrt{3}}{2})$, $(-\frac{1}{2}, -\frac{\sqrt{3}}{2})]^T$ for a equilateral hexagon of side $a$. Note that the scaling coefficient for a planner formation $a$ could be either constant or time varying, e.g., $a=Bsin(\omega t)$. In the transformed domain the formation basis is defined as $\xi_d=[\xi_{sd}^T,\xi_{cd}]^T$, where, $\xi_{sd}=[\xi_{1d},...,\xi_{5d}]^T=\Phi_sz_d$ with $\Phi=[\Phi_s^T,\Phi_c^T]^T$ as follows
\begin{equation}
\label{phi}
\Phi=\begin{bmatrix}
\Phi_s \\ \hline \Phi_c
\end{bmatrix}=\begin{bmatrix}
\frac{-1}{\sqrt{2}} & \frac{1}{\sqrt{2}} & 0 & 0 & 0 & 0 \\
0 & 0 & \frac{1}{\sqrt{2}} & \frac{-1}{\sqrt{2}} & 0 & 0 \\
0 & 0 & 0 & 0 & \frac{1}{\sqrt{2}} & \frac{-1}{\sqrt{2}} \\
\frac{-1}{{2}} & \frac{-1}{{2}} & \frac{1}{{2}} & \frac{1}{{2}} & 0 & 0 \\
\frac{1}{{4}} & \frac{1}{{4}} & \frac{1}{{4}} & \frac{1}{{4}} & \frac{-1}{{2}} & \frac{-1}{{2}} \\ \hline
\frac{1}{{6}} & \frac{1}{{6}} & \frac{1}{{6}} & \frac{1}{{6}} &\frac{1}{{6}} & \frac{1}{{6}}
\end{bmatrix}
\end{equation}
$\xi_{cd}=t+\imath sint$ is the desired trajectory of the centroid. It can be checked that the leading principle minors of $\Phi^{-1}$ are nonzero and the eigenvalues are not all positive. Then, by theorem, there exists a stabilizing matrix $D=[-1.4140-1.4140\iota,-2.4498+2.4498\iota,-4.6189-2.3095\iota,-3.7244+1.8622\iota,- 1.6117\iota,0.0340]^T$ such that the eigenvalues of $D\Phi^{-1}$ have positive real parts. Therefore the $6$ agents asymptotically reach the given[] planner formation. A simulation of a moving formation with $6$ agents is shown in Fig.
\begin{figure}[h]
\begin{center}
\includegraphics[width=8cm, height=6cm]{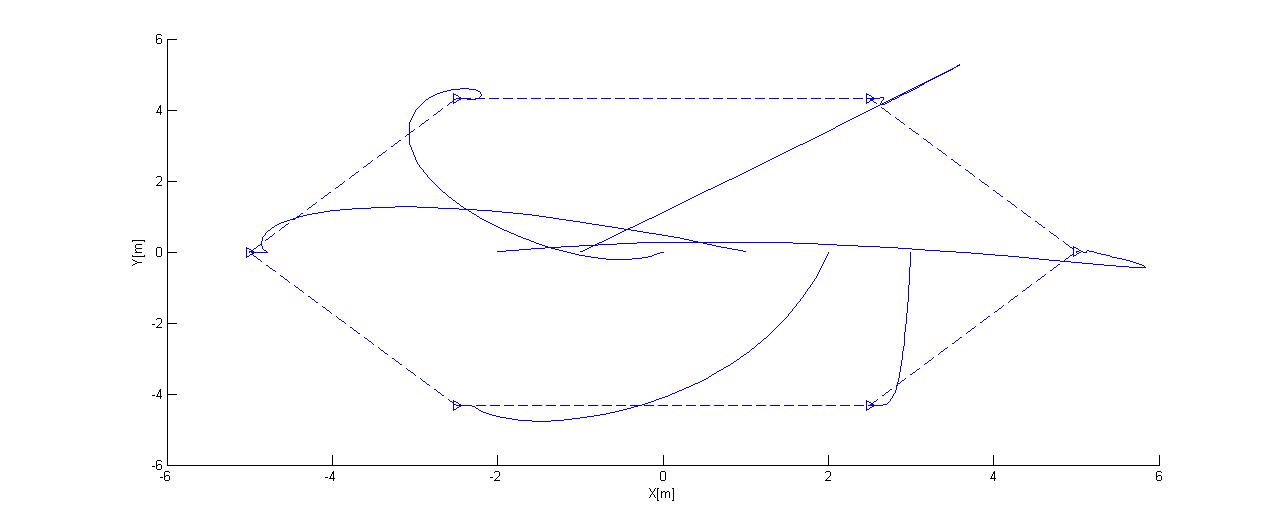}    
\caption{Schematic representation of multiple groups of robots}  
\label{fig1}                                 
\end{center}                                 
\end{figure}
\begin{figure}[h]
\begin{center}
\includegraphics[width=8cm, height=6cm]{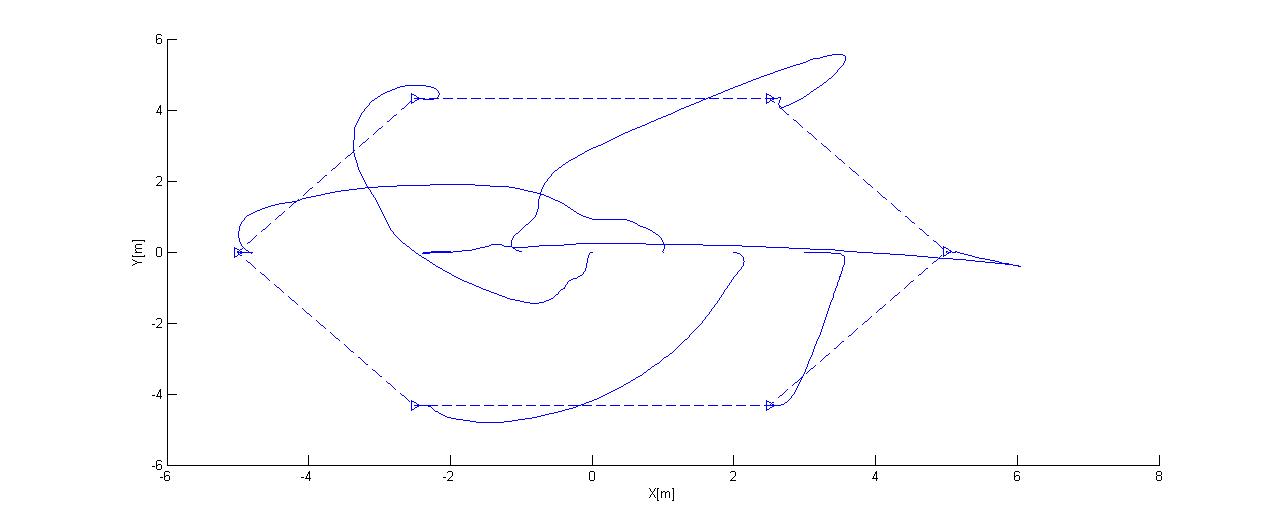}    
\caption{Schematic representation of multiple groups of robots}  
\label{fig1}                                 
\end{center}                                 
\end{figure}

\section{Conclusion}
This paper addresses two different aspects of planner formation for a set of complex single integrator systems. Therefore the dimension of the system reduces to half. The first problem is a stabilization problem associated with a real linear transformation and the second one is to find a sophisticated way to carry out stability analysis in the transformed domain for the system subjected to collision avoidance.  Instead of choosing the gain directly, we write the system of equations in the transformed domain such that the inverse of transformation appears in the closed loop dynamics. This way the stabilization problem is formulated. The inverse of transformation must have nonzero eigenvalues with both positive and negative real parts. This may lead to instability of the system. Therefore, the inverse of transformation is pre-multiplied by a diagonal stabilizing matrix so that all the eigenvalues of the product is reassigned with positive real parts. Hence the system becomes stable. Then we define \textit{matrix of potential} and formulate a matrix of potential in the transformed domain using the matrix of potential in actual domain to get the same effect of collision avoidance in the transformed domain. This means that we give the structure of collision avoidance controller in the transformed domain. Therefore we do not need to add the collision avoidance controller going back to the actual domain from the transformed domain after designing a stabilizing controller in the transformed domain. We provided a systematic way to find a diagonal stabilizing matrix. Simulation results are given to support our claim.


\section*{Acknowledgment}

The authors would like to thank...



\bibliographystyle{IEEEtran}
\bibliography{mybib}
%
%
%

\end{document}